\documentclass[a4paper, 10 pt, conference]{ieeeconf}
\usepackage{cite}
\usepackage{graphicx}
\usepackage{graphics}
\usepackage{color}
\usepackage{colortbl}
\usepackage{xcolor}
\usepackage[utf8]{inputenc}
\usepackage[english]{babel}
\usepackage{amsmath}

\usepackage{amsthm}
\usepackage{amssymb}  
\usepackage{setspace}
\let\labelindent\relax
\usepackage{enumitem}
\usepackage{centernot}
\usepackage{bm}
\newtheorem{theorem}{Theorem}
\IEEEoverridecommandlockouts                              

\begin{document} 

\title{
Mobility and Safety Benefits of Connectivity in CACC Vehicle Strings
}
%
%
%

\author{Vamsi~Vegamoor, Shaojie~Yan, Sivakumar~Rathinam,~\IEEEmembership{Senior~Member,~IEEE}
        and~Swaroop~Darbha,~\IEEEmembership{Fellow,~IEEE.}
\thanks{Vamsi Vegamoor is a doctoral student  in Mechanical Engineering at Texas A\&M University, College Station, TX 77843, USA; e-mail: vvk@tamu.edu}
\thanks{Sivakumar Rathinam and Swaroop Darbha are faculty members in Mechanical Engineering at Texas A\&M University, College Station, TX 77843-3123, USA; e-mail: \{srathinam, dswaroop\}@tamu.edu}
}

\maketitle

\begin{abstract}
In this paper, we re-examine the notion of string stability as it relates to safety by providing an upper bound on the maximum spacing error  of any vehicle in a homogeneous platoon in terms of the input of the leading vehicle. We reinforce our previous work on lossy CACC platoons by accommodating for burst-noise behavior in the V2V link. Further, through Monte Carlo type simulations, we demonstrate that connectivity can enhance traffic mobility and safety in a CACC string even when the deceleration capabilities of the vehicles in the platoon are heterogeneous. 
\end{abstract}

\section{Introduction}
Connected and Autonomous vehicles (CAVs) have received renewed interest from the automotive industry and policy makers alike \cite{Brookings}. Connected vehicles have the potential to improve safety and mobility on roadways by increasing throughput as they enable vehicle following with short inter-vehicle spacing that would otherwise be unsafe for human drivers to maintain. Adaptive Cruise Control (ACC) systems use on-board sensors to measure relative distance and relative velocity to the preceding vehicle. Cooperative Adaptive Cruise Control (CACC) systems have the added capability to obtain acceleration of one or more preceding vehicles (among other information) by utilizing Vehicle to Vehicle (V2V) communication.

To prevent collisions in a platoon of vehicles, local fluctuations in spacing errors need to be damped out as they propagate across the string of vehicles. This condition for string stability is often expressed and analysed in terms of a frequency domain condition \cite{swaroop1994phd} which ensures that the 2-norm of spacing errors do not amplify. However from a safety perspective, the maximum spacing error of the vehicles is more relevant as it dictates if a collision will occur. In this work, we present an upper bound for this infinity norm of spacing errors as a function of the lead vehicle's input signal.

For studying string stability, we model vehicles as point masses whose acceleration can be controlled through first order actuation dynamics:
\begin{equation}
    \ddot{x}_i=a_i, \quad \tau \dot{a}_i+a_i=u_i,
\end{equation}{}
where $u_i$ is the control input. While the lag in individual vehicles of the platoon may vary, it is reasonable to assume that it is bounded above by some value for all vehicles. We will use this upper bound as the maximum lag $\tau$ for the platoon. In this way, any heterogeneity in the parasitic lags in the platoon can be accounted for. We note that such models have been used successfully in experiments conducted in the California PATH projects \cite{CalPathExample,CalPathExample2}.

It has been long established \cite{swaroop1994phd} that for an ACC platoon, string stability can be guaranteed if the time headway chosen is at least twice the sum of parasitic lags in the vehicle, assuming homogeneity in vehicle capabilities. For passenger cars, these actuation lags are typically in the range of $0.2s$ to $0.5s$, necessitating a minimum time headway in the range of $0.4s$ to $1s$. This translates to maintaining an inter-vehicle spacing of $12$ to $30$ metres at highway speeds. Heavy-duty vehicles may have larger actuation lags, requiring spacing greater than $30 \;m$ to operate under ACC. This is not efficient, both from traffic mobility and fuel savings perspectives, as we know that shorter spacing is required to take advantage of slip-streaming effects for reduced aerodynamic drag. 


Previously, it was demonstrated that the time headway can be safely reduced further in CACC platoons using V2V communication\cite{V2VBenefits}. These results were initially developed on the assumption that the V2V link is lossless. But in reality, wireless channels are prone to packet drops due to interference or bandwidth restrictions. Earlier work from the authors \cite{vegamoor2019} proposed a new limit on the minimum time headway for CACC platoons, given a packet reception probability. In that work, the packet transmission process over the V2V link was modelled as a binomial random process.  It was brought to the authors' attention that packet drops over a transmission channel are not i.i.d (independent and identically distributed) in practice and that burst-noise models have been more widely accepted \cite{PacketLossRev2}. The discuss in Section \ref{sec:lossyCACC} reinforces our previous result by proving that the same limit is applicable for bursty packet drops.

Finally, we also discuss the benefits of connectivity by exploring the statistics of collisions in ACC and CACC platoons. String stability can be guaranteed if the time headway is chosen appropriately, provided the commanded acceleration/deceleration during platooning operation is within the maximum capabilities of each vehicle. But if the control input exceeds the capabilities of some of vehicles in the string, then string stability guarantees do not hold. Through numerical simulations, we demonstrate a reduction in the probability and number of collision events in heterogeneous CACC platoons compared to heterogeneous ACC platoons.

In short, the contributions of this work are as follows:
\begin{itemize}
    \item Provide a bound on the maximum spacing error of any vehicle in a platoon, which has more bearing on safety compared to traditional requirements of string stability.
    \item Extend the authors' previous result for lossy CACC platoons by accounting for burst noise characteristics in the transmission channel.
    \item Demonstrate the benefits of connectivity for reducing the probability of collisions in heterogeneous traffic.
\end{itemize}

\section{Maximum Spacing Errors in a String}\label{sec:SpacingBound}
Consider a string of $N$ vehicles, where vehicles are indexed in an ascending order with index $0$ referring to the lead vehicle. Let $\zeta_i(t)$ denote the state of the $i^{th}$ vehicle in a string at time $t$; $y_i(t)$ denote the output of the $i^{th}$ vehicle (such as spacing and velocity errors in the $i^{th}$ vehicle with respect to some origin).  Let $d_i(t)$ be the disturbance acting on the $i^{th}$ vehicle.  Let ${\cal S}_i$ denote the set of vehicles whose information is available to the $i^{th}$ vehicle for feedback. Let ${\mathcal I}_N:=\{1,2, \ldots, N\}$ denote the set of indices of all the vehicles in the platoon except the lead vehicle. For some appropriate functions $f_{ij}$ and $h_i$, the evolution of spacing errors may be described by a set of equations of the form:
\begin{eqnarray*}
{\dot \zeta}_i = \sum_{j \in S_i} f_{ij}(\zeta_i, \zeta_j, d_i), \quad e_i = h_i(\zeta_i), \quad i \in {\mathcal I}_N.
\end{eqnarray*}

When the disturbances are absent, note that $\zeta_i =0, \; i \in {\mathcal I}_N$ is an equilibrium solution of the above set of coupled evolution equations. We use the following generalization of the definition of string stability due to Ploeg et al \cite{PVN2014}, Besselink and Knorn\cite{BK2018}:

\vspace*{0.05in}

\noindent{\bf Definition (Scalable Weak Input-State Stability):} The nonlinear system is said to be scalably input-output stable if there exist functions $\beta \in {\mathcal KL}$ and $\sigma \in {\mathcal K}$ and a number $N_{min}$ such that for any $N \ge N_{min}$ and for any bounded disturbances $d_i(t), \; i \in {\mathcal I}_N$, 
 $$\max_{i \in {\mathcal I}_N}\|\zeta_i(t)\| \le \beta(\sum_{i \in {\mathcal I}_N}\|\zeta_i(0)\|, t) + \sigma(\max_{i \in {\mathcal I}_N} \|d_i(t)\|_{\infty}).$$

With feedback linearization, these equations reduce to:
\begin{eqnarray}
\label{eq:Lin}
\dot \zeta_1 &=& A_0 \zeta_1 + D w_0, \\
    \dot \zeta_i &=& A_0 \zeta_i + B y_{i-1}, \quad \forall i \ge 2\\
    y_i &=& C \zeta_i, \quad \forall i \ge 1,
\end{eqnarray}
where $w_0(t)$ denotes the acceleration of the lead vehicle,
$A_0$ is Hurwitz matrix, $B, C, D$ are respectively constant matrices.

In applications such as Adaptive Cruise Control (ACC) and Cooperative Adaptive Cruise Control (CACC), the set of vehicles from which information is available is ${\cal S}_i = \{i-1\}$ for a single preceding vehicle lookup scheme. In the next generation of CACC systems (CACC+ systems), we could have ${\cal S}_i = \{i-1, i-2, \ldots, i-r\},$ where $r$ depends on the connectivity. In a string of identical vehicles as has been shown in \cite{swaroop1994phd, V2VBenefits}, one obtains the following error evolution equations using a Laplace transformation for the case ${\cal S}_i = \{i-1\}$:
$$Y_i(s) = H(s) Y_{i-1}(s), $$
where $H(s)$ is a rational, proper, stable transfer function. The requirement of string stability has thus far \cite{Shahab90, ioannou1993autonomous, swaroop1994phd} been used as $\|H(jw)\|_{\infty} \le 1$.\\

From \cite{desoerfeedback}, it is known that the input-output relationship for a rational, proper transfer function is:
$$\|y_i\|_2 \le \|H(jw)\|_{\infty}\|y_{i-1}\|_2, $$
where the input and output are measured by their ${\mathcal L}_2$ norms (power in the error signals). Practical consideration for this application requires us to consider $\|y_i\|_{\infty}$ (the maximum value of the output) as it has direct bearing on safety; however, the corresponding input-output relationship from \cite{desoerfeedback} is
$$\|y_i\|_{\infty} \le \|h(t)\|_1 \|y_{i-1}\|_{\infty}, $$
where $h(t)$ is the unit impulse response of the transfer function $H(s)$. It is known from \cite{desoerfeedback} that $H(0)\le \|H(jw)\|_{\infty} \le \|h(t)\|_1$ and that $H(0) = \int_0^{\infty} h(t) = \|h(t)\|_1$, when $h(t) \ge 0$ for all $t \ge 0$. Typical information flow structures such as the one for ACC and CACC are such that $H(0)=1$, thereby putting a lower bound on $\|h(t)\|_1 =1$. However, for ascertaining string stability, one must attain this lower bound; an obstacle to attaining the lower bound is to find controller gains that render the unit impulse response of $H(s)$ non-negative. This is a variant of the open problem of transient control and there are currently no systematic procedures for determining the set of gains for this case.    

In this article, we exploit the bounded structure of leader's acceleration and the finite duration of lead vehicle maneuvers to prove that it suffices to consider $\|H(jw)\|_{\infty} \le 1$ to show the {\it uniform} boundedness of spacing errors.
\begin{theorem} \label{thm:outputBound}
Suppose:
\begin{itemize}
    \item The  error propagation equations are given by 
    \begin{eqnarray}
    \dot \zeta_1(t) &=& A_0 \zeta_1(t) + D w_0(t), \\
    \dot \zeta_i(t) &=& A_0 \zeta_i(t) + B y_{i-1}(t), \quad \forall i \ge 2\\
    y_i(t) &=& C \zeta_i(t), \quad \forall i \ge 1,
    \end{eqnarray}
    and $A_0$ is a Hurwitz matrix; 
    \item the lead vehicle executes a bounded acceleration maneuver in finite time, i.e., $w_0(t) \in {\cal L}_2 \cap {\cal L}_{\infty}$;
    \item $\|C(jw I - A_0)^{-1}B\|_{\infty} \le 1$ and
    \item For some $\alpha^*>0$, $\sum_{i=1}^N \|\zeta_i(0)\| \le \alpha^*$ for every $N$. 
\end{itemize}{} 
Then, there exists a $M_1, M_2 >0$, independent of $N$, such that for all $i \ge 1$:
$$\|y_{i}(t)\|_{\infty} \le M_1 \alpha^* +M_2\|w_0(t)\|_2. $$
\end{theorem}
\begin{proof}
From $A_0$ being Hurwitz, and from Linear System Theory \cite{desoerfeedback}, one obtains for some constants, $\beta_2, \beta_{\infty}, \gamma_2, \gamma_{\infty}$, 
\begin{align*}
\zeta_1(t) &= e^{A_0t}\zeta_1(0) + \int_0^t e^{A_0(t-\tau)}D w_0(\tau) d \tau, \\
\zeta_i(t) &= e^{A_0t} \zeta_i(0) + \int_0^t e^{A_0(t-\tau)}B y_{i-1}(\tau) d \tau, \;  i \ge 2, \\
\Rightarrow \|y_1(t)\|_2 &\le \beta_2 \|\zeta_1(0)\| + \gamma_2\|w_0(t)\|_2, \\
\|y_1(t)\|_{\infty} &\le \beta_{\infty}\|\zeta_1(0)\| + \gamma_{\infty}\|w_0(t)\|_{\infty}, \\
\|y_i(t)\|_2 &\le \beta_2 \|\zeta_i(0)\| +  \|y_{i-1}(t)\|_2, \; \; i \ge 2.
\end{align*}
Note that the last inequality results from $ \|C(jwI-A_0)^{-1}B\|_{\infty} \le 1$. The last two inequalities can be expressed as:

\begin{align*}
\|y_{i}(t)\|_2 &\le \beta_2 (\sum_{j=2}^i \|\zeta_j(0)\|) + \|y_1(t)\|_2, \\
&\le \beta_2 (\sum_{i\in {\cal I}_N} \|\zeta_i(0)\|) + \gamma_2 \|w_0(t)\|_2, \\
&\le \beta_2 \alpha^* + \gamma_2 \|w_0(t)\|_2.
\end{align*}

From \cite{corless-zhu-skelton}, it follows that if 
\begin{eqnarray*}
J^* &:=& \max \quad trace (C^TPC), \quad \textit{subject to} \\
P &\succ& 0, \quad AP + PA^T+BB^T = 0,
\end{eqnarray*}
then for some $\eta >0$ and for all $i\ge 1$,
\begin{eqnarray*}
\|y_{i}(t)\|_{\infty} &\le&  \eta \|\zeta_i(0)\| + J^* \|y_{i-1}(t)\|_2, \\
&\le& (J^*\beta_2+\eta) \alpha^* + J^* \gamma_2\|w_0(t)\|_2.
\end{eqnarray*}

This completes the proof.
\end{proof}
\noindent{\bf Remark:} The condition that the initial errors must be absolutely summable is trivially satisfied as there are only finitely many vehicles in a string. For guaranteeing that errors are within a specified bound, one must ensure that the absolute sum of initial errors is within acceptable levels. 



\section{Adjusting Time Headway in Presence of Packet Drops in CACC vehicle strings} \label{sec:lossyCACC}
\subsection{CACC vehicle strings}
A typical constant time headway control law for the $i^{th}$ following vehicle in an ideal, loss-less CACC string can be written as:
\begin{align}
        u_i=K_{a}a_{i-1}-K_{v}(v_i-v_{i-1})-K_{p}(x_i-x_{i-1}+h_wv_i),
\end{align}
where $h_w$ is the time headway, ($K_a,K_v,K_p$) are tunable gains, $u_i$ is the control input and $x_i$, $v_i$, $a_i$  are states of the $i^{th}$ vehicle. 

Majority of the previous analytical work on CACC systems assumed that the communication channels are lossless. It was understood from experiments that losses in the V2V link would degrade vehicle following capability and may induce instability \cite{LeiITSconf2011,LossyCommsVargas}. Works like \cite{GracefulCACC} attempt to overcome this problem by estimating the lost information as necessary. Until recently \cite{vegamoor2019}, there were no quantitative bounds directly relating the minimum string stable time headway to the loss characteristics of the channel.  We restate the key results from \cite{vegamoor2019} for convenience:
\begin{itemize}[noitemsep,nolistsep]
    \item By modeling packet reception over the V2V link as a binomial random process, we showed that if a large number of realizations are averaged, the state trajectories of the stochastic system converge to that of a deterministic system where the random parameters in the state space representation are replaced by their expectations.
    \item Using this deterministic equivalent system, we derived a lower bound on the time headway as a sufficient condition for string stability:
\end{itemize}
\begin{align}\label{eqn:ECCResult}
h_w\ge h_{min}=\frac{2\tau}{1+\gamma K_a}, 
\end{align}
where $\gamma$ is the probability of successfully receiving a packet - a quantity that can be updated in real time using simple network performance measurement tools. Also, we have made the assumption that whenever a packet is successfully received, it contains accurate acceleration information. In reality, the information encapsulated may have sensor noise; we plan to account for this in our future work.

\subsection{Incorporating the Gilbert Channel Model}\label{subSec:GilbertCACC}
The Gilbert model \cite{Gilbert1960} and some of its extensions \cite{Gilbert_Eliott} \cite{ExtendedGilbert} are extensively used to simulate bursts of noise that occur in wireless transmission channels. The Gilbert model consists of two states: a `Good' state where no packets are corrupted/dropped, and a `Bad' state, where only $q\%$ of the packets are transmitted error-free. Let the transition probabilities from `Good' to `Bad' and `Bad' to `Good' be $P$ and $Q$ respectively. The transition diagram for a communication link between $i^{th}$ and $(i-1)^{th}$ vehicle is shown below, where $\hat{w} \in \{1,0\}$. 
 \begin{figure}[h]
 \centering
 \includegraphics[width=0.45\textwidth]{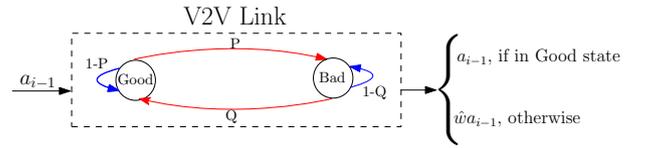}
 \caption{Communication link with Gilbert model}
 \label{fig:GilbertFig}
 \end{figure}

Since errors only occur in the bad state, the probability of a dropped packet is:
\begin{align}
    \mathbb{P}(\hat{w}=0)=(1-q)\frac{P}{P+Q}.
\end{align}
Consequently, the expectation of $\hat{w}$ is:
\begin{align}\label{eqn:GilbertGamma}
    \mathbb{E}[\hat{w}]=1-\frac{P(1-q)}{P+Q} = :\gamma,
\end{align}
where $\gamma$ is the probability of successful packet reception.

$P$ and $Q$ are typically small, if the states are to persist. Moreover, we have assumed that it is possible to return from the `Bad' state to the `Good' state.  If, for example, a hardware fault occurs and it is not possible for the V2V link to return to the `Good' state, then $Q=0$ and we will continue to receive packets with the probability $q$. This would be equivalent to the situation presented in our earlier work \cite{vegamoor2019}, with $\gamma =q$. 

Consider a platoon of $k$ vehicles. The $i^{th}$ following vehicle obtains the acceleration of the $(i-1)^{th}$ vehicle through wireless communication. Random variables $\hat{w}_{i,j} \in \{1,0\}$ are used to represent the reception/loss of the acceleration packet from the $j^{th}$ vehicle to the $i^{th}$ vehicle. From measurement, we can obtain $\gamma:= \mathbb{E}[\hat{w}_{i,j}]$. Without loss of generality, we can consider $\gamma$ to be same for the whole platoon. Let the lead vehicle be imparted some control action $u_L$ by a driver (or otherwise).
The equation of motion for the lead vehicle and each of the $i^{th}$ following vehicles, $i\ge1$ is given by:
\begin{align} \label{eqn:StochasticCACC}
    \tau \dot{a}_0+a_0&=u_L, \nonumber \\
    \tau \dot{a}_i+a_i&=u_i=\hat{w}_{i,i-1}K_{a}a_{i-1}-K_{v}(v_i-v_{i-1}) \nonumber \\
    & \qquad \qquad-K_{p}(x_i-x_{i-1}+h_wv_i).
\end{align}
For the remainder of this section, we work under the following assumptions which are reasonable from a practical perspective.
\begin{itemize}
    \item The leading vehicle's trajectory is purely deterministic.
    \item The V2V link operates at a rate equal to or greater than the vehicle controller's sampling rate.  
    \item The communication link between any pair of vehicles is independent from any other pair. That is, the state of one transceiver doesn't affect the state of other transceivers in the platoon
\end{itemize}

If we consider the platoon of vehicles as a stochastic system, its equation of motion can be written as:
\begin{align}\label{eqn:StochasticStateSpace}
    \dot{\hat{X}}=\hat{A}(\hat{w}(t))\hat{X}+BU,
\end{align}
where $\hat{X}$=$(x_0,v_0,a_0,x_1,v_1,a_1,\cdots x_k,v_k,a_k)$ and $U =u_L$, the input to the lead vehicle. Note that only the system matrix $\hat{A}(\hat{w}(t))$ has random elements.

Let $\Delta t$ be the controller time step so that the total (finite) run time is $t_m =m\Delta t$, $m\in \mathbb{N}$. Let us consider the evolution of the stochastic state vector over the first time interval $[0,t_1)$:
\begin{align}
    \hat{X}(t_1)= \hat{\Phi}(t_1,0)\hat{X}(0)+\int_{0}^{t_1}\hat{\Phi}(t_1,\zeta)BU(\zeta)d\zeta \label{eqn:StochasticSTM},
\end{align} 
where ${\hat{\Phi}(t_1,0)}$ is the stochastic state transition matrix, dependent on the values of $\hat{w}_{i,j}$ at $t=0$. For small controller time steps, it is reasonable to assume that the input $U$ is updated by the controller at the beginning of each time step and is held constant during that interval.
\begin{align}
    \hat{X}(t_1)= \hat{\Phi}(t_1,0)\hat{X}(0)+\int_{0}^{t_1}\hat{\Phi}(t_1,\zeta)d\zeta BU(0)\label{eqn:StochasticSTM_adjusted}
\end{align} 
Since we have defined $\mathbb{E}[\hat{w}_{i,j}]=\gamma$, let us consider replacing the random elements in the system matrix of equation (\ref{eqn:StochasticStateSpace}) with their expected values. Then we get some deterministic system:
\begin{align}
        \dot{\bar{X}}=\bar{A}\bar{X}+BU
\end{align}
Our goal now is to show that $\mathbb{E}[\dot{\hat{X}}(t)]=\dot{\bar{X}}(t)$, for all $t \in [0,t_m]$.
For the deterministic system, the state evolution for the first interval $[0,t_1)$ is:
\begin{align}\label{eqn:DeterministicSTM}
     \bar{X}(t_1)= \bar{\Phi}(t_1,0)\bar{X}(0)+\int_{0}^{t_1}\bar{\Phi}(t_1,\zeta)d\zeta BU(0).
\end{align}
Now consider $\hat{\Phi}(t_1,0)$ and $\bar{\Phi}(t_1,0)$. Since $\hat{A}(\hat{w}(t))$ only changes at each controller time step, it is constant in the interval $[0,t_1)$ and takes the value $\hat{A}(\hat{w}(0))=:\hat{A}_1$. So, we can write
\begin{align}
        \hat{\Phi}(t_1,0)&=e^{\int_0^{t_1} \hat{A}(\hat{w}(\xi))d\xi}= e^{\hat{A}_1t_1} \label{eqn:BasePhiHat}\\
        \bar{\Phi}(t_1,0)&=e^{\int_0^{t_1} \bar{A}d\xi}= e^{\bar{A}t_1}
\end{align}

Now we use the power series expansion for the exponential matrices:
\begin{align}
    e^{\hat{A}_1t_1}&=I+\hat{A}_1t_1+\frac{(\hat{A}_1t_1)^2}{2!}+\frac{(\hat{A}_1t_1)^3}{3!} +\cdots \label{eqn:expAhat}\\ 
    e^{\bar{A}t_1}&=I+\bar{A}t_1+\frac{(\bar{A}t_1)^2}{2!}+\frac{(\bar{A}t_1)^3}{3!} +\cdots \label{eqn:expAbar}
\end{align}
While generally not true for random matrices \cite{matrixConvexity},
\begin{equation}\label{eqn:MatrixExpectation2}
\mathbb{E}[\hat{A}^n_1]= \bar{A}^n \text{ }\forall n \in \mathbb{N}
\end{equation}
for CACC system matrices with one vehicle lookup.

This is due to its specific structure since the diagonal elements of the system matrix are purely deterministic and the powers of $\hat{A}$ only contain elements that are multi-linear in $\hat{w}_{i,j}$. 
This allows us to exploit the fact that the expectation of a product of independent random variables is the product of their expectations. We have noticed that this convenient multi-linear property of the powers of system matrices is afforded only for one vehicle lookup schemes (CACC) but not for platoons that utilize communicated information from two or more preceding vehicles (CACC+ systems). An example with system matrices for a three vehicle platoon has been provided in the arXiv version of this paper \cite{arXiv_Version}, which has been omitted here for brevity.

Thus, over a large number of realizations,
    \begin{align}\label{eqn:BasePhiExp}
        \mathbb{E}[\hat{\Phi}(t_1,0)]=\bar{\Phi}(t_1,0). 
    \end{align}
Since the initial conditions can be assumed to be the same in equations (\ref{eqn:StochasticSTM}) and (\ref{eqn:DeterministicSTM}), i.e., $\hat{X}(0)=\bar{X}(0)$, we get:
\begin{align}
    \mathbb{E}[\hat{X}(t_1)]=\bar{X}(t_1),
\end{align} \label{eqn:InductionBase}
for the first interval $[0,t_1)$. Let this form the base case with the induction hypothesis for interval $[t_{k-1},t_k)$ as:
\begin{align}
    \mathbb{E}[\hat{X}(t_k)]=\bar{X}(t_k) \label{eqn:InductionHypth}
\end{align}
Now consider the next interval $[t_k,t_{k+1})$
\begin{align*}
    \hat{X}(t_{k+1})= \hat{\Phi}(t_{k+1},t_k)\hat{X}(t_k)+\int_{t_k}^{t_{k+1}}\hat{\Phi}(t_{k+1},\zeta)d\zeta BU(t_k) \\
    \bar{X}(t_{k+1})= \bar{\Phi}(t_{k+1},t_k)\bar{X}(t_k)+\int_{t_k}^{t_{k+1}}\bar{\Phi}(t_{k+1},\zeta)d\zeta BU(t_k)
\end{align*}
Using a similar reasoning as in equations (\ref{eqn:BasePhiHat} - \ref{eqn:BasePhiExp}), we can show that $\mathbb{E}[\hat{\Phi}(t_{k+1},t_k)]=\bar{\Phi}(t_{k+1},t_k)$.

Again, note that the term $\hat{\Phi}(t_{k+1},t_k)\hat{X}(t_k)$ only contains products of independent random variables. From the induction hypothesis in equation (\ref{eqn:InductionHypth}), we can claim $\mathbb{E}[\hat{\Phi}(t_{k+1},t_k)\hat{X}(t_k)]=\bar{\Phi}(t_{k+1},t_k)\bar{X}(t_k)$. This yields:
\begin{align}
    \mathbb{E}[\hat{X}(t_{k+1})]=\bar{X}(t_{k+1}).
\end{align}
From the principle of mathematical induction, $\mathbb{E}[\hat{X(t)}]=\bar{X}(t)$ for all finite $t \in [0,t_m]$. This allows us to replace equation (\ref{eqn:StochasticCACC}) with its deterministic equivalent.
\begin{align}
    \tau \dot{a}_i+a_i&=\gamma K_{a}a_{i-1}-K_{v}(v_i-v_{i-1}) \nonumber \\ & \qquad-K_{p}(x_i-x_{i-1}+d+h_w v_i) 
\end{align}
Following the procedure in \cite{V2VBenefits} for this governing equation, we obtain the same bound on the minimum employable time headway as in our earlier work \cite{vegamoor2019}.
\begin{align} \label{eqn:PL_CACC_Limit}
h_w\ge h_{min}=\frac{2\tau}{1+\gamma K_a} 
\end{align}
\subsection{Numerical Simulations with Gilbert Channel}
Let us consider a homogeneous platoon where the parasitic lags of all vehicles are upper bounded by $\tau=0.5s$. The transition probabilities from Fig. \ref{fig:GilbertFig} were set to $P=0.3$ and $Q=0.1$. Further, we assume that the all packets are transmitted successfully while the channel is in the `Good' state and only $20\%$ of the packets are successfully transmitted in the `Bad' state (i.e., $q=0.2$). 
We now simulate a platoon of six (one lead + five following) vehicles operating under a constant time headway policy as stated in equation (\ref{eqn:StochasticCACC}) using Simulink. The lead vehicle initially moves with a constant velocity of $25 m/s$, then at $t=10s$, decelerates at the rate of $-9m/s^2$ to $16 m/s$, which it maintains for the rest of the simulation. This setup simulates an emergency braking maneuver. Gains ($K_a, K_v, K_p$) were set to ($0.4,1,0.8$). Spacing error plots for the first, third and fifth following vehicle for a time headway of $0.75s$ are shown in Fig. \ref{fig:GilbertUnstableSim}.
 \begin{figure}[h] 
 \centering
 \includegraphics[width=5.6cm, height=4cm]{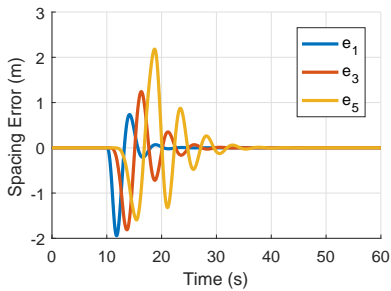}
 \caption{Spacing errors of a platoon using $h_w=0.75s$ (Unstable).}
\label{fig:GilbertUnstableSim}
 \end{figure}  

We see that the string is unstable as the spacing errors are amplified across the platoon. This is expected, since even though the headway chosen ($0.75s$) is above the ideal, lossless CACC limit\cite{V2VBenefits} ($\frac{2\tau}{1+K_a} =0.71s$), the lossy Gilbert channel induces instability in the platoon. To accommodate for this, we calculate the average packet reception rate, $\gamma$ from equation (\ref{eqn:GilbertGamma}) to be used in equation (\ref{eqn:PL_CACC_Limit}). We get $\gamma=0.4$ and $h_{min}=0.86s$.

So we repeat the simulation with a new headway $h_w=0.9>h_{min}$. From Fig. \ref{fig:GilbertStableSim}, we observe that the spacing errors diminish as expected. It is to be noted that the time headway employed is still below the minimum stable headway for ACC platoons ($2\tau=1s$), showing that even with lossy communication links there is no need to downgrade from CACC to ACC as long as the headway is adjusted using equation (\ref{eqn:PL_CACC_Limit}).

 \begin{figure}[h]
 \centering
 \includegraphics[width=5.6cm,height=4cm]{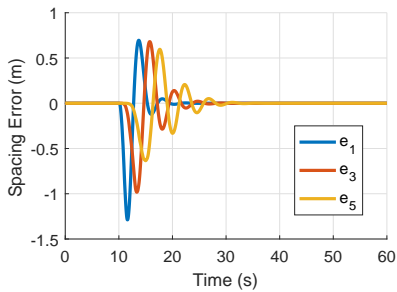}
 \caption{Spacing errors of a platoon using $h_w=0.9s$ (Stable).}
  \label{fig:GilbertStableSim}
 \end{figure}
\section{Safety Analysis with ACC and CACC strings} \label{sec:Safety}

So far, we have only considered homogeneous platoons. In reality, not all vehicles on roadways are identical. There are two sources of heterogeneity that need to be accounted for from a safety perspective: heterogeneity in the time headway selected by each vehicle and heterogeneity in maximum deceleration limits of vehicles.

In most commercial Advanced Driver Assistance Systems (ADAS), the time headway for vehicle following is user-selectable, leading to heterogeneity in time headway. Such platoons have been analysed previously \cite{OroszHeteroDelays,NijmeijerHeteroCACC}. Moreover, in \cite{Ankem2019EffectOH}, is shown that even with a random distribution of time headway selection, string stability can be maintained with appropriate gain selection.

In this section, we study the second source of homogeneity: deceleration limits of each vehicle. A vehicle's ability to brake is affected by various factors (tire size, tire wear, road surface, etc.). This leads to a distribution in the maximum achievable deceleration of each vehicle in a platoon. Consequently, a rear end collision could occur if a vehicle with worn/old brakes follows another with a better deceleration capability. 

Safety in such scenarios can be assessed using metrics such as probability of collision and average number of collision events per platoon (a measure of severity of instability in the platoon). We perform a large number of numerical simulations for ACC and CACC platoons under emergency braking and compare these metrics between the two platoon modes. The six (1 leader + 5 followers) vehicle simulations were set up using the following guidelines:

\begin{itemize}
    \item At beginning of every realization, maximum deceleration for each vehicle is randomly picked based on a probability distribution provided by Godbole and Lygeros (1997)\cite{GodboleLygeros97} for passenger vehicles.
    \item  Initially, all vehicles are at equilibrium and are travelling at $30 m/s$. The lead vehicle brakes with its maximum deceleration assigned by probability distribution from $t=0$ till it comes to a complete stop.
    \item Each vehicle is assumed to be a rigid body with a constant length. When the position of any points of two vehicles overlap, we consider a collision to have occurred. If two vehicles collide, both of them are assumed to stop instantaneously. This is a source of potential improvement in future work as the current simulation setup does not account for momentum transfer effects.
\end{itemize}

For these simulations, the parameters ($K_p, K_v, K_a, \tau $) were selected as ($2,0.8,0.25, 0.5$) respectively and time headway was picked as $1s$. Note that the time headway is sufficient to guarantee string stability for both ACC and CACC platoons if all vehicles were homogeneous. Communication for CACC platoons are assumed perfect with no packet drops.  At end of each realization, we record whether a collision occurred (platoon is not string stable) and if so, the number of collision events involved. This was repeated to obtain 10,000 realizations each for ACC and CACC platoons. The probability of collision and average numbers of collision events per unstable platoon is shown in Fig. \ref{BarFig} for ACC and CACC.
 \begin{figure}[h]
  \centering
 \includegraphics[width=5cm,height=4cm]{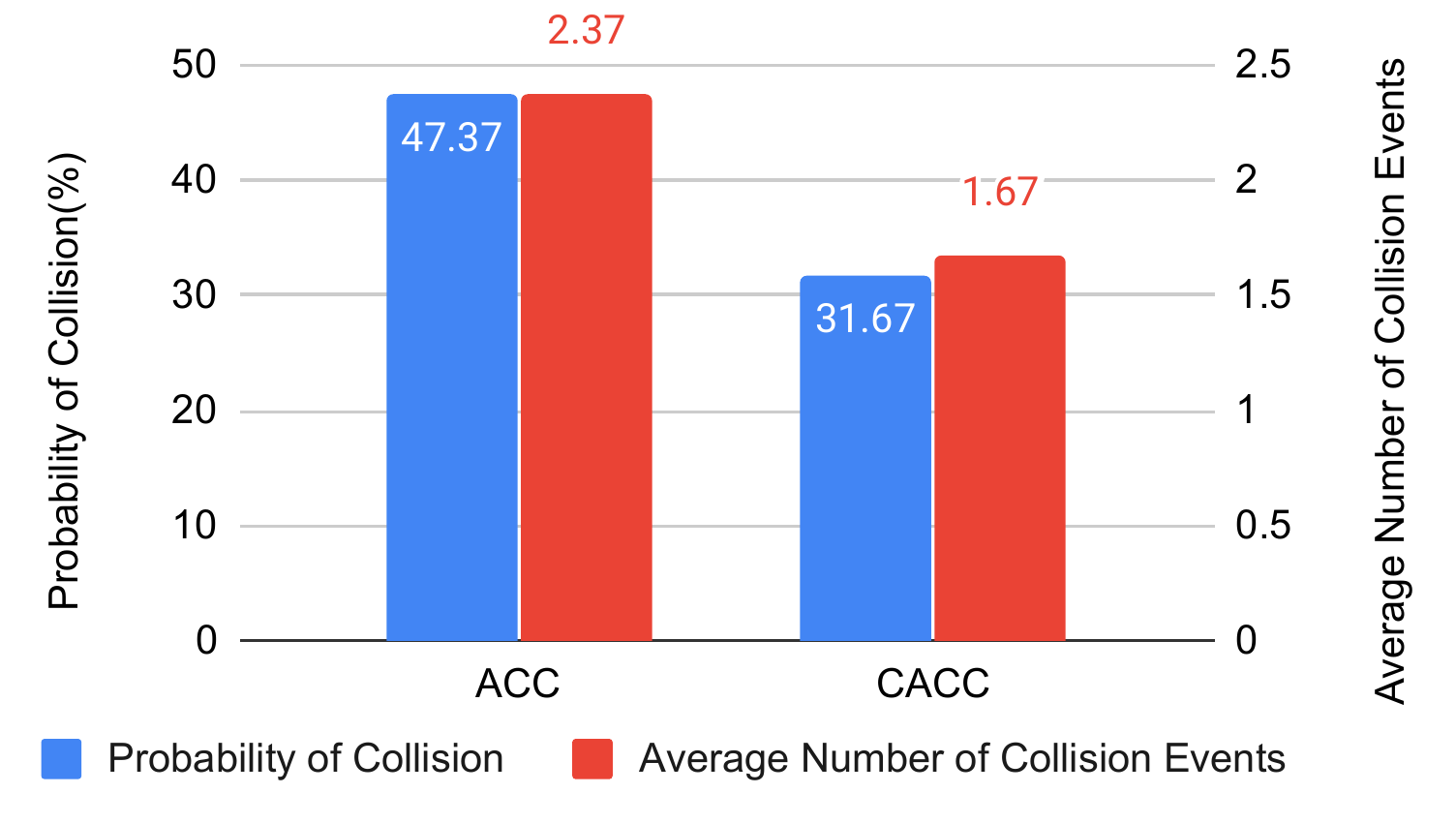}
 \caption{Probability and Severity of collisions.}
 \label{BarFig}
 \end{figure}
 
 We observe that both the probability and severity of the collisions decrease in heterogeneous platoons when a CACC control law is implemented, showing the benefits of connectivity from a safety perspective. A similar reduction was observed in the variance of spacing errors over time in the CACC platoon, as shown in Fig. \ref{VarianceFig}.
  \begin{figure}[h]
 \centering
 \includegraphics[width=9.5cm]{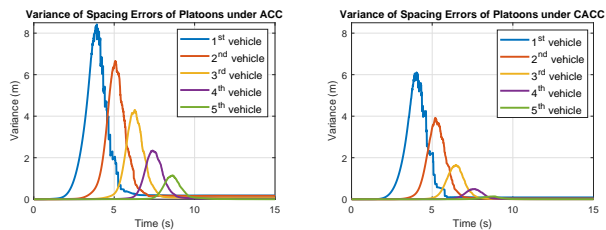}
 \caption{Variance of spacing errors over time}
 \label{VarianceFig}
 \end{figure}
 \vspace{-0.3cm}
 \section{Conclusion}

 We have developed a method to identify the maximum perturbation in position for any vehicle in a platoon given the lead vehicle's input which is relevant for safety. For CACC vehicle strings prone to packet drops, we presented a sufficient condition on the time headway to guarantee string stability and demonstrated its validity for burst-noise channels. Apart from experimental validation using drive-by-wire capable passenger cars, future work in this direction may involve exploring lossy CACC+ systems which are not amenable to the same analysis presented for lossy CACC systems in this paper. Finally, we discussed the improvements to safety metrics afforded by connectivity in vehicle strings with heterogeneous deceleration capabilities. 

 Acknowledgment: This work was partially supported by US Department of Transportation (FHWA) through grant number 693JJ32045024.

\bibliographystyle{IEEEtran}
\bibliography{Journal}

\begin{thebibliography}{10}
\providecommand{\url}[1]{#1}
\csname url@samestyle\endcsname
\providecommand{\newblock}{\relax}
\providecommand{\bibinfo}[2]{#2}
\providecommand{\BIBentrySTDinterwordspacing}{\spaceskip=0pt\relax}
\providecommand{\BIBentryALTinterwordstretchfactor}{4}
\providecommand{\BIBentryALTinterwordspacing}{\spaceskip=\fontdimen2\font plus
\BIBentryALTinterwordstretchfactor\fontdimen3\font minus
  \fontdimen4\font\relax}
\providecommand{\BIBforeignlanguage}[2]{{%
\expandafter\ifx\csname l@#1\endcsname\relax
\typeout{** WARNING: IEEEtran.bst: No hyphenation pattern has been}%
\typeout{** loaded for the language `#1'. Using the pattern for}%
\typeout{** the default language instead.}%
\else
\language=\csname l@#1\endcsname
\fi
#2}}
\providecommand{\BIBdecl}{\relax}
\BIBdecl

\bibitem{Brookings}
{Brookings Institution}, ``Autonomous cars: Science, technology and policy,''
  \url{https://www.brookings.edu/events/autonomous-cars-science-technology-and-policy/},
  July 2019.

\bibitem{swaroop1994phd}
S.~Darbha, ``String stability of interconnected systems: An application to
  platooning in automated highway systems,'' PhD Dissertation, University of
  California Berkeley, 1994.

\bibitem{CalPathExample}
F.~Bu, H.~Tan, and J.~Huang, ``Design and field testing of a cooperative
  adaptive cruise control system,'' in \emph{Proceedings of the 2010 American
  Control Conference}, June 2010, pp. 4616--4621.

\bibitem{CalPathExample2}
\BIBentryALTinterwordspacing
R.~Rajamani, S.~B. Choi, B.~K. Law, J.~K. Hedrick, R.~Prohaska, and P.~Kretz,
  ``{Design and Experimental Implementation of Longitudinal Control for a
  Platoon of Automated Vehicles },'' \emph{Journal of Dynamic Systems,
  Measurement, and Control}, vol. 122, no.~3, pp. 470--476, 06 1998. [Online].
  Available: \url{https://doi.org/10.1115/1.1286682}
\BIBentrySTDinterwordspacing

\bibitem{V2VBenefits}
S.~{Darbha}, S.~{Konduri}, and P.~R. {Pagilla}, ``Benefits of {V2V}
  communication for autonomous and connected vehicles,'' \emph{IEEE
  Transactions on Intelligent Transportation Systems}, vol.~20, no.~5, pp.
  1954--1963, May 2019.

\bibitem{vegamoor2019}
V.~K. {Vegamoor}, D.~{Kalathil}, S.~{Rathinam}, and S.~{Darbha}, ``Reducing
  time headway in homogeneous {CACC} vehicle platoons in the presence of packet
  drops,'' in \emph{2019 18th European Control Conference (ECC)}, June 2019,
  pp. 3159--3164.

\bibitem{PacketLossRev2}
C.~A. G.~D. {Silva} and C.~M. {Pedroso}, ``{MAC}-layer packet loss models for
  wi-fi networks: A survey,'' \emph{IEEE Access}, vol.~7, pp.
  180\,512--180\,531, 2019.

\bibitem{PVN2014}
J.~{Ploeg}, N.~{van de Wouw}, and H.~{Nijmeijer}, ``Lp string stability of
  cascaded systems: Application to vehicle platooning,'' \emph{IEEE
  Transactions on Control Systems Technology}, vol.~22, no.~2, pp. 786--793,
  March 2014.

\bibitem{BK2018}
B.~{Besselink} and S.~{Knorn}, ``Scalable input-to-state stability for
  performance analysis of large-scale networks,'' \emph{IEEE Control Systems
  Letters}, vol.~2, no.~3, pp. 507--512, July 2018.

\bibitem{Shahab90}
S.~{Sheikholeslam} and C.~A. {Desoer}, ``Longitudinal control of a platoon of
  vehicles,'' in \emph{1990 American Control Conference}, May 1990, pp.
  291--296.

\bibitem{ioannou1993autonomous}
P.~A. Ioannou and C.~C. Chien, ``Autonomous intelligent cruise control,''
  \emph{IEEE Transactions on Vehicular Technology}, vol.~42, no.~4, pp.
  657--672, 1993.

\bibitem{desoerfeedback}
\BIBentryALTinterwordspacing
C.~Desoer and M.~Vidyasagar, \emph{Feedback systems: input-output properties},
  ser. Electrical science series.\hskip 1em plus 0.5em minus 0.4em\relax
  Academic Press, 1975. [Online]. Available:
  \url{https://books.google.com/books?id=C\_1SAAAAMAAJ}
\BIBentrySTDinterwordspacing

\bibitem{corless-zhu-skelton}
M.~{Corless}, G.~{Zhu}, and R.~{Skelton}, ``Improved robustness bounds using
  covariance matrices,'' in \emph{Proceedings of the 28th IEEE Conference on
  Decision and Control,}, Dec 1989, pp. 2667--2672 vol.3.

\bibitem{LeiITSconf2011}
C.~Lei, M.~van Eenennaam, W.~K. Wolterink, and J.~Ploeg, ``Impact of packet
  loss on cacc string stability performance,'' in \emph{Proceedings of the 11th
  Intl. Conference on ITS Telecommunications}, 2011.

\bibitem{LossyCommsVargas}
F.~J. Vargas, A.~I. Maass, and A.~A. Peters, ``String stability for predecessor
  following platooning over lossy ommunication channels,'' in \emph{23rd
  International Symposium on Mathematical Theory of Networks and
  Systems}.\hskip 1em plus 0.5em minus 0.4em\relax Hong Kong: Hong Kong
  University of Science and Technology, July 2018.

\bibitem{GracefulCACC}
J.~{Ploeg}, E.~{Semsar-Kazerooni}, G.~{Lijster}, N.~{van de Wouw}, and
  H.~{Nijmeijer}, ``Graceful degradation of cooperative adaptive cruise
  control,'' \emph{IEEE Transactions on Intelligent Transportation Systems},
  vol.~16, no.~1, pp. 488--497, 2015.

\bibitem{Gilbert1960}
\BIBentryALTinterwordspacing
E.~N. Gilbert, ``Capacity of a burst-noise channel,'' \emph{Bell System
  Technical Journal}, vol.~39, no.~5, pp. 1253--1265, 1960. [Online].
  Available:
  \url{https://onlinelibrary.wiley.com/doi/abs/10.1002/j.1538-7305.1960.tb03959.x}
\BIBentrySTDinterwordspacing

\bibitem{Gilbert_Eliott}
E.~O. {Elliott}, ``Estimates of error rates for codes on burst-noise
  channels,'' \emph{The Bell System Technical Journal}, vol.~42, no.~5, pp.
  1977--1997, Sep. 1963.

\bibitem{ExtendedGilbert}
\BIBentryALTinterwordspacing
H.~A. Sanneck and G.~Carle, ``{Framework model for packet loss metrics based on
  loss runlengths},'' in \emph{Multimedia Computing and Networking 2000},
  K.~Nahrstedt and W.~chi Feng, Eds., vol. 3969, International Society for
  Optics and Photonics.\hskip 1em plus 0.5em minus 0.4em\relax SPIE, 1999, pp.
  177 -- 187. [Online]. Available: \url{https://doi.org/10.1117/12.373520}
\BIBentrySTDinterwordspacing

\bibitem{matrixConvexity}
\BIBentryALTinterwordspacing
M.~H. Rizvi and R.~W. Shorrock, ``A note on matrix-convexity,'' \emph{The
  Canadian Journal of Statistics / La Revue Canadienne de Statistique}, vol.~7,
  no.~1, pp. 39--41, 1979. [Online]. Available:
  \url{http://www.jstor.org/stable/3315013}
\BIBentrySTDinterwordspacing

\bibitem{OroszHeteroDelays}
\BIBentryALTinterwordspacing
W.~B. Qin and G.~Orosz, ``Scalable stability analysis on large connected
  vehicle systems subject to stochastic communication delays,''
  \emph{Transportation Research Part C: Emerging Technologies}, vol.~83, pp. 39
  -- 60, 2017. [Online]. Available:
  \url{http://www.sciencedirect.com/science/article/pii/S0968090X17301869}
\BIBentrySTDinterwordspacing

\bibitem{NijmeijerHeteroCACC}
C.~{Wang} and H.~{Nijmeijer}, ``String stable heterogeneous vehicle platoon
  using cooperative adaptive cruise control,'' in \emph{2015 IEEE 18th
  International Conference on Intelligent Transportation Systems}, Sep. 2015,
  pp. 1977--1982.

\bibitem{Ankem2019EffectOH}
M.~D. Ankem and S.~Darbha, ``Effect of heterogeneity in time headway on error
  propagation in vehicular strings,'' \emph{2019 IEEE Intelligent
  Transportation Systems Conference (ITSC)}, pp. 2612--2619, 2019.

\bibitem{GodboleLygeros97}
D.~N. Godbole and J.~Lygeros, ``Tools for safety-throughput analysis of
  automated highway systems,'' in \emph{Proceedings of the 1997 American
  Control Conference (Cat. No.97CH36041)}, vol.~3, Jun 1997, pp. 2031--2035
  vol.3.

\end{thebibliography}

\end{document}